\documentclass[12pt]{amsart}

\usepackage{amsmath}
\usepackage{amscd,amssymb}
\usepackage{epsfig}

\usepackage{graphicx}
\usepackage{color}
\usepackage{hyperref}
\usepackage{booktabs}
\usepackage{amsfonts}
\usepackage{mathrsfs}
\usepackage{bm}

\def\qed{\ifmmode\square\else\nolinebreak\hfill
$\Box$\fi\par\vskip12pt}

\newtheorem{thm}{Theorem}[section]
\newtheorem{lemma}[thm]{Lemma}

\newtheorem{proposition}[thm]{Proposition}
\numberwithin{equation}{section}
\numberwithin{thm}{section}

\theoremstyle{definition}

\newtheorem{remark}[thm]{Remark}
\newtheorem{example}[thm]{Example}

\newcommand{\bD}{\mathbb D}\newcommand{\bF}{\mathbb F}

\newcommand{\cB}{\mathcal B}\newcommand{\cC}{\mathcal C}

\newcommand{\cP}{\mathcal P}\newcommand{\cR}{\mathcal R}

\definecolor{Purple}{rgb}{0.5,0,0.5}

  \def\Tr{{\rm Tr}}

\newcommand{\ra}{\rightarrow}

\newcommand{\ol}{\overline}

\newcommand{\<}{\langle}
\renewcommand{\>}{\rangle}

\begin{document}
\pagestyle{plain}
\parindent=19pt

\title{Linear codes of $2$-designs as subcodes of the extended generalized Reed-Muller codes}
\begin{center}
\author{Zhiwen He$^{a}$. Jiejing Wen$^{b,c,*}$
}\end{center}
\address{School of Mathematical Sciences, Zhejiang University, Hangzhou 310027,  China a}
\email{zhiwen$\_$he@zju.edu.cn}

\address{Key Laboratory of Cryptologic Technology and Information Security, Ministry of Education, Shandong University, Qingdao 266237, China b}
\address{School of Cyber Science and Technology, Shandong University, Qingdao 266237, China c}
\email{jjwen@sdu.edu.cn}
\begin{abstract} 
This paper is concerned with the affine-invariant ternary codes which are defined by Hermitian functions. We compute the incidence matrices of 2-designs that are supported by the minimum weight codewords of these ternary codes. The linear codes generated by the rows of these incidence matrix are subcodes of the extended codes of the 4-th order generalized Reed-Muller codes and they also hold 2-designs. Finally, we give the dimensions and lower bound of the minimum weights of these linear codes.
\end{abstract}

\keywords{ternary code, $2$-design, incidence matrix, generalized Reed-Muller code.\\
{\bf  Mathematics Subject Classification 94B15. 05B05. 51E10}\\
{\bf  Funding information: National Natural Science Foundation of China under Grant No. 11771392.}\\
$^*$Correspondence author}

\maketitle


\section{Introduction}
A $t$-design with parameters $(v,k,\lambda)$ is a pair $\mathbb{D}=(\cP,\cB)$ with point set $\cP$ and block set $\cB$, where $\cP$ has size $v$ and each block in $\cB$ is a $k$-subset of $\cP$, such that any $t$ points are contained in $\lambda$ blocks. We only consider simple design, that is a design contains no repeated blocks, with $v>k>\lambda$. Let $q$ be a prime power and $\bF_q$ be a finite field. An $[n,k,d]$ linear code $\cC$ over $\bF_q$ is a $k$-dimensional linear subspace of vector space $\bF_q^n$ with minimum distance $d$. Let $A_i$, $0\leq i\leq n$, denote the number of codewords of weight $i$ in $\cC$. The sequence $(A_0,A_1,\cdots,A_{n})$ is called weight distribution of $\cC$ and $\sum_{i=0}^{n}A_it^i$ is weight enumerator of $\cC$.

The theories of $t$-design and linear code are closely related. Let $\bD$ be a $t$-$(v,k,\lambda)$ design. Let $b$ be the number of blocks in $\cB$. The incidence matrix $M_{\bD}=(m_{ij})$ of $\bD$ is a $b\times v$ matrix where $m_{ij}=1$, if $p_j$ is in $B_i$ and $m_{ij}=0$, otherwise.  The rows of the incidence matrix $M_{\bD}$ can be viewed as vectors of $\bF_q^v$. Then the subspace $\cC(\mathbb{D})$ spanned by these $b$ vectors is called the linear code of $\bD$ over $\bF_q$ of length $n=v$. Let $\cC$ be an $[n,k,d]$ linear code over $\bF_q$ with each codeword indexed by the ordered elements $\{p_0,p_1,\cdots,p_{n-1}\}$.  For any $A_i\neq 0$, let $\cB_i$ be the collection of the supports $\textup{Suppt}(c)=\{p_j:c_j\neq 0,0\leq j\leq n-1\}$ for all $c=(c_0,c_1,\cdots,c_{n-1})\in \mathcal{C}$ with weight $i$, $0\leq i\leq n$. Let $\cP=\{p_0,p_1,\cdots,p_{n-1}\}$. If the pair $(\cP,\cB_i)$ is a $t$-$(v,k,\lambda)$ design with $v=n, k=i$ for some positive integers $\lambda$ and $t\geq 2$, then we call it the support design of the code $\cC$ and denote it by $\mathbb{D}_i(\cC)$.

A number of papers demonstrated that $t$-designs and codes are interesting topics for combinatorics and coding theorists: see \cite{DC1} for a summary of constructive results in $t$-design and linear code. Many infinite families of $2$-designs and $3$-designs have been constructed from codes in different methods, for instance \cite{DC3}$-$\cite{DX2}, \cite{WR}.  But there are only a few examples of $4$-designs and $5$-designs which are obtained from the Golay binary and ternary codes. Recently, the $71$-year-old open problem of the existence of infinity families of linear codes holding $4$-designs is settled by Tang and Ding in \cite{TC} and it remains open whether there exist infinity families of linear codes holding $t$-design with $t\geq5$. Ding, Tang and Tonchev \cite{DC2} studied the linear codes of $2$-designs held in a class of affine-invariant ternary codes. The ternary codes used in their paper are defined by quadratic functions. In this paper, we consider the affine-invariant ternary codes which are defined by Hermitian functions. The linear codes $\cC_3(\bD_d(\cC(2m,3)))$ obtained in this paper are affine-invariant which means that they also hold $2$-designs. Moreover, the new linear codes $\cC_3(\bD_d(\cC(2m,3)))$ have the original codes $\cC(2m,3)$ as its subcodes and it has many other affine-invariant subcodes. This implies that the structure of $\cC_3(\bD_d(\cC(2m,3)))$ is richer than the previous one. We simplify the linear code $\cC_3(\mathbb{D}_d(\cC(2m,3)))$ and get that it is a subcode of the extended code of the $4$-th order generalized Reed-Muller code.

The organization of this paper is as follows. In Section \ref{sec_2}, we introduce some basic knowledge of cyclic codes, generalized Reed-Muller codes, and automorphism group of codes. In Section \ref{sec_3}, we consider the designs that are held in affine-invariant ternary codes and the linear codes that are spanned by the incidence matrices of these designs. We present the generators of the linear code $\cC_3(\bD_d(\cC(2m,3)))$ in Theorem \ref{thm_3.15} and state the dimension and the lower bound of the minimum weight of $\cC_3(\bD_d(\cC(2m,3)))$ in Theorem \ref{thm_3.16}. In Section \ref{sec_4}, we give the proof of the main results that are given in Section \ref{sec_3}. In Section \ref{sec_5}, we conclude this paper.

\section{Preliminaries}\label{sec_2}
\subsection{Cyclic code} 
An $[n,k,d]$ linear code $\cC$ over $\bF_q$ is cyclic code if for each codeword $c=(c_0,c_1,\cdots,c_{n-1})\in \cC$ the shift codeword $(c_1,c_2,\cdots,c_{0})$ is also in $\cC$. We define a residue class ring $\cR_n[x]=\bF_q[x]/(x^n-1)$ and a subset $\cC(x)$ of $\cR_n[x]$ corresponding to the cyclic code $\cC$ \[\cC(x)=\{c_0+c_1x+\cdots+c_{n-1}x^{n-1}\in \cR_n[x]: (c_0,c_1,\cdots,c_{n-1})\in \cC\}.\] There is a bijection between the cyclic code $\cC$ and the subset $\cC(x)$ of $\cR_n(x)$. It is easy to see that $xc(x)\in \cC(x)$ for any $c(x)\in \cC(x)$. Hence $\cC(x)$ forms an ideal in the residue class ring $\cR_n[x]$.  Since $\cR_n[x]$ is a principal ideal domain, $\cC(x)$ is principal and $\cC(x)=\<g(x)\>$ for some monic polynomial $g(x)\in \cR_n[x]$.  We call $g(x)$ the generator polynomial and $h(x)=(x^n-1)/g(x)$ the parity-check polynomial of $\cC$.

Let $n$ be an integer such that gcd$(n,q)=1$. The $q$-cyclotomic coset $C_s$, $0\leq s<n$, of $s$ modulo $n$ is defined by \[
C_s=\{s,sq,\cdots,sq^{r-1}\}(\textup{mod } n),\]where $r$ is the smallest positive integer such that $sq^r\equiv s(\textup{mod } n)$. Note that these distinct $q$-cyclotomic cosets partition the set $\{0,1,\cdots,n-1\}$. Let $m$ be the order of $q$ module $n$, $\xi$ be the primitive element of $\bF_{q^m}$, i.e. $\bF_{q^m}^*=\left \langle \xi \right \rangle$. For each $s$, $0\leq s< n$, the minimal polynomial of $\xi^s$ over $\bF_p$ is $M_{\xi^s}(x)=\Pi_{i\in C_s}(x-\xi^i)$. The generator polynomial can be write as $g(x)=\Pi_{s\in T}M_{\xi^s}(x)$, where $T=\cup_s C_s$ is the union of some $q$-cyclotomic cosets. We call $T$ the defining set of the cyclic code $\cC$. The roots of unity in $Z=\{\xi^i: i\in T\}$ are called zeros of the cyclic code $\cC$ and the roots of unity in $\bF_{q^m}^*\setminus Z$ are nonzeros of $\cC$. We refer readers to \cite{HW} for more details on cyclotomic cosets and minimal polynomials.

The dual code $\cC^{\perp}$ of the cyclic code $\cC$ is defined by \[\cC^{\perp}:=\{c'\in \bF_q|c\cdot c'=0,\textup{ for any }c\in \cC\},\]where $\cdot$ is inner product. We know that $\cC^{\perp}$ has the generator polynomial $x^nh(x^{-1})$.

The following theorem from \cite{HW} shows that the zeros of $\cC^{\perp}$ can be derived from $\cC$.
\begin{thm}[See Thm.4.4.9 \cite{HW}]\label{thm_2.1}
Let $\cC$ be an $[n,k,d]$ cyclic code over $\bF_q$ with generator $g(x)$ and $h(x)=(x^n-1)/g(x)$ be its parity-check polynomial. If $\alpha_1,\cdots,\alpha_k$ are the nonzeros of $\cC$, then $\alpha_1^{-1},\cdots,\alpha_k^{-1}$ are the zeros of $\cC^{\perp}$.
\end{thm}

\begin{proposition}[\cite{HW}]\label{lem_2.2}
Let $\cC_i$ be cyclic codes of length $n$ over $\bF_q$ with defining sets $T_i$ for $1\leq i\leq k$. Then the linear code $\cC_1+\cC_2+\cdots+\cC_k=\{c_1+\cdots+c_k\ |\ c_i\in \cC_i,1\leq i\leq k\}$ has defining set $T_1\cap T_2\cap\cdots\cap T_k$.
\end{proposition}

For any codeword $c=(c_0,c_1,\cdots,c_{n-1})\in \cC$ we adjoin an extra letter $c_n$ such that $c_0+c_1+\cdots+c_{n}=0$. In this way we can get the extended code $\ol{\cC}$ of $\cC$. If $H$ is the parity check matrix of $\cC$, then the parity check matrix of $\ol{\cC}$ is \[\ol{H}=\begin{bmatrix}\bm{1}&1\\H&\bm{0}\end{bmatrix},\]where $\bm{1}=(1,1,\cdots,1)$ and $\bm{0}=(0,0,\cdots,0)^{\top}$. 
\begin{thm}[See Thm.4.2.1 \cite{HW}]\label{thm_2.3}
Let $\cC$ be a nonzero cyclic code and $g(x)$ be its general polynomial. Then the dimension of $\cC$ is $n-deg(g(x))$.
\end{thm}

\begin{thm}[See Thm.2.8 \cite{DC1}]\label{thm_2.4}
Let $\cC$ be an $[n,k,d]$ code over $\bF_q$. Then $\ol{\cC}$ is an $[n+1,k,\ol{d}]$ linear code, where $\ol{d}=d$ if $d$ is even, and $\ol{d}=d+1$, otherwise.
\end{thm}

\begin{thm}[See Thm.4.4.19 \cite{HW}]\label{thm_2.5}
Let $n$ be a positive integer and $q$ be a prime power. Let $g(x)$ be an irreducible factor of $x^n-1$ over $\bF_q$. Suppose $g(x)$ has degree $s$, and let $\gamma\in\bF_{q^s}$ be the root of $g(x)$. Let $\Tr_s:\bF_{q^s}\ra\bF_q$ be the trace map from $\bF_{q^s}$ to $\bF_q$. Then \[\cC_{\gamma}=\{\sum_{i=0}^{n-1}\Tr_s(a\gamma^i)x^i|a\in\bF_{q^s}\}\]is the $[n,s]$ irreducible cyclic code with nonzeros $\{\gamma^{-q^i}|0\leq i<s\}$.
\end{thm}

\subsection{The generalized Reed-Muller codes}
Let $q$ be a prime power and $l,m$ be positive integers with $1\leq l\leq (q-1)m$. An $l$-th order punctured generalized Reed-Muller code $\cR_q(l,m)^*$ over $\bF_q$ is the cyclic code of length $n=q^m-1$ with generator polynomial \[
g(x)=\sum_{\overset{1\leq i\leq n-1}{w_q(i)<(q-1)m-l}}(x-\xi^i),\] where $\xi$ is the primitive element of $\bF_{q^m}$, $i=\sum_{j=0}^{m-1}i_jq^j$ with $0\leq i_j\leq q-1$ and $w_q(i)=\sum_{i=0}^{m-1}i_j$.

Assmus and Key \cite{As} provided the parameters of the punctured generalized Reed-Muller code $\cR_q(k,m)^*$ in the following theorem.
\begin{thm}[See Thm.5.24 and Cor.5.26\cite{As}]\label{thm_2.6}
The code $\cR_q(l,m)^*$ has dimension \[k=\sum_{i=0}^l\sum_{j=0}^m(-1)^j\binom{m}{j}\binom{i-jq+m-1}{i-jq}\]and minimum weight $d=(q-l_0)q^{m-l_1-1}-1$, where $l=l_1(q-1)+l_0$ and $0\leq l_0\leq q-1$.
\end{thm}

 The dual of the punctured generalized Reed-Muller code $(\cR_q(k,m)^*)^{\perp}$ and its parameters are obtained in \cite{As} and \cite{As0}.
\begin{thm}[See Cor.5.21 \cite{As}]\label{thm_2.7}
The code $(\cR_q(l,m)^*)^{\perp}$ is the cyclic code with generator polynomial \[g^{\perp}(x)=\prod_{\overset{1\leq j\leq n-1}{\omega_q(j)\leq l}}(x-\alpha^j).\]
\end{thm}

\begin{thm}[See Sect.5.4 \cite{As0}]\label{lem_2.8}
The code $(\cR_q(l,m)^*)^{\perp}$ has length $n=q^m-1$, dimension \[k^{\perp}=n-\sum_{i=0}^l\sum_{j=0}^m(-1)^j\binom{m}{j}\binom{i-jq+m-1}{i-jq},\]and minimum weight \[d^{\perp}\geq (q-l_0’)q^{m-l_1’-1},\]where $m(q-1)-1-l=l_1’(q-1)+l_0’$ and $0\leq l_0’<q-1$.
\end{thm}

\subsection{Automorphism group of linear code}
Let $\cC$ be an $[n,k,d]$-linear code over $\bF_q$. Let $\textup{PAut}(\cC)$ be the set of coordinate permutations that maps $\cC$ to itself. Note that the set $\textup{PAut}(\cC)$ forms a subgroup of the symmetric group $\textup{Sym}(n)$. For any $a\in\bF_q^n$, we define a linear transformation on $\bF_q^n$ by \[\phi_a(c)=(a_0c_0,\cdots,a_{n-1}c_{n-1})\textup{ for each }c=(c_0,\cdots,c_{n-1})\in\cC.\] The monomial group $\textup{MAut}(\cC)$ is the set of linear transformations on $\bF_q^n$ that preserves the Hamming weight, in which the element has the form of $\sigma\phi_a$ with $\sigma\in\textup{PAut}(\cC)$ and $a\in \bF_q^n$ such that $\phi_a$ leaves $\cC$ invariant. The automorphism group $\textup{Aut}(\cC)$ is the set of maps $\sigma\phi_a\gamma$, where $\sigma\phi_a\in \textup{MAut}(\cC)$, $\gamma$ is a field automorphism of $\bF_q$ that leaves $\cC$ invariant, and $\sigma\phi_a\gamma$ acts on $\cC$ by\[\sigma\phi_a\gamma(c)=(\gamma(a_0c_{\sigma(0)}),\gamma(a_1c_{\sigma(1)}),\cdots,\gamma(a_{n-1}c_{\sigma(n-1)})) \textup{ for each }c\in \cC.\]

The automorphism group $\textup{Aut}(\cC)$ is called $t$-transitive if for any two $t$-tuples $(c_1,\cdots,c_t)$ and $(c_1’,\cdots,c_t’)$ with $c_i,c_i'\in\bF_q$, $1\leq i\leq t$, there exists an element $\sigma\phi_a\gamma\in \textup{Aut}(\cC)$ such that $(c_1^{\sigma},\cdots,c_t^{\sigma})=(c_1',\cdots,c_t')$.

The following lemma is the sufficient condition for a linear code $\cC$ to hold $t$-design. 
\begin{thm}[See Thm.8.4.7 \cite{HW}]\label{lem_2.9}
Let $\cC$ be a code if length $n$ over $\bF_q$. If $\textup{Aut}(\cC)$ is $t$-transtive, then the codewords of any weight $i\geq t$ of $\cC$ hold a $t$-design.
\end{thm}

Let $\cC$ be a $[n,k,d]$-linear code indexed by the elements of $\bF_q$ with $n=q$. The general affine group $\textup{GA}_1(\bF_q)$ is the set of permutations of $\bF_q$:\[\{\sigma_{s_1,s_2}: s_1\in \bF_q^*,s_2\in\bF_q\},\] where $\sigma_{s_1,s_2}(x)=s_1x+s_2$ for any $x\in\bF_q$. The linear code $\cC$ is called affine invariant if the general affine group $\textup{GA}_1(\bF_q)$ leaves $\cC$ invariant. Actually, an affine-invariant code is an extended cyclic code such that $\textup{GA}_1(\bF_q)\subseteq\textup{PAut}(\cC)$. The following theorem demonstrate that affine-invariant is a very useful property to determine which extended cyclic code holds $2$-design.

\begin{thm}[See Thm.6.6 \cite{DC1}]\label{thm_2.10}
Let $A_i$ be the number of codeword of weight $i$ for $0\leq i\leq n$. If the linear code $\cC$ is affine invariant, then for each $i$ with $A_i\neq 0$, the supports of the codewords of weight $i$ in $\cC$ form a $2$-design.
\end{thm}

\section{Codes of designs held in a class of affine-invariant ternary code}\label{sec_3}
Let $m\geq 2$ be an positive integer, $p$ be an odd prime and $q=p^{2m}$. Let $\Tr_{s}$ be the trace map from $\bF_{p^s}$ to $\bF_p$. Let $n=q-1$. We consider the linear code  
\begin{equation}\label{Eqn_TC}
\cC(2m,p)=\{c(a,b,h)\ |\  a\in \bF_{p^m},b\in \bF_{p^{2m}},h\in \bF_p\},
\end{equation}
where 
\[
c(a,b,h)=(\Tr_{2m}(at^{p^m+1}+bt)+h)_{t\in \bF_{p^{2m}}}.
\]

As noted in \cite{DX2}, the code $\cC(2m,p)$ is affine invariant, thus it holds $2$-designs. 
For each codeword $c(a,b,h)$ in $\cC(2m,p)$, the Hamming weight $w_H(c(a,b,h))=p^{2m}-T(a,b,h)$, where 
\begin{equation}\label{Eqn_T}
T(a,b,h)=|\{t\in\bF_q\ |\ \Tr_{2m}(at^{p^m+1}+bt)+h=0\}|.
\end{equation}

\begin{lemma}[\cite{DX2}]
Let $T(a,b,h)$ be defined in $(\ref{Eqn_T})$ for $a\in \bF_{p^m}$, $b\in \bF_q$ and $h\in\bF_p$. Then 
\begin{itemize}
\item[(1)] If $a=b=h=0$, then $T(a,b,h)=p^{2m}$.\\
\item[(2)] If $a=b=0$ and $h\ne 0$, then $T(a,b,h)=0$.\\
\item[(3)] If $a=0$ and $b\ne 0$, then $T(a,b,h)=p^{2m-1}$.\\
\item[(4)] If $a\ne 0$, then \[T(a,b,h)=\begin{cases}p^{2m-1}-p^{m-1}(p-1)\textup{ if }h=\textup{Tr}(as^{p^m+1}_{at,bt})\\p^{2m-1}+p^{m-1}\quad\quad\quad\textup{ if }h\ne \Tr(as^{p^m+1}_{at,bt})\end{cases},\]where $t\in \bF_p^*$ and $s_{at,bt}^{p^m+1}$ is a solution of $((at)^{p^m}+at)s=2ats=-(bt)^{p^m}$, i.e. $s_{at,bt}^{p^m+1}=-2^{-1}a^{-1}b^{p^m}t^{p^m-1}=a^{-1}b^{p^m}$.
\end{itemize}
\end{lemma}

Then each codeword $c(a,b,h)$ has minimum weight $d=p^{2m-1}(p-1)-p^{m-1}$ only if $a\in\bF_{p^m}$, $b\in \bF_{p^{2m}}$ and $h\neq \Tr(b)$. The linear code $\cC(2m,p)$ has parameters $[p^{2m},3m+1,p^{2m-1}(p-1)-p^{m-1}]$ from Theorem 3 in \cite{DX2}. Let $\bD_d(\cC(2m,p))$ be a design $(\cP,\cB)$, in which the blocks formed by the supports of codewords in $\cC(2m,p)$ with minimum weight $d$. We know that $\bD_d(\cC(2m,p))$ is a $2$-design from \cite{DX2}. Let $M_{\bD_d}$ be the incidence matrix of $\bD_d(\cC(2m,p))$ and $\cC_p(\bD_d(\cC(2m,p)))$ be a linear code spanned by the row vectors of $M_{\bD_d}$ over $\bF_p$. We will restrict ourselves to the case of $p=3$ and try to compute the dimension and minimum weight of $\cC_3(\bD_d(\cC(2m,3)))$.

\begin{table}[htbp]
       \centering
       \caption{The weight distribution of $\cC(2m,p)$}
       \label{table}
       \begin{tabular}{cccc}
             \hline
             &\\[-6pt]
             Weight&Multiplicity\\
             \hline
             &\\[-6pt]
             $0$&$1$\\
             &\\[-6pt]
             $p^{2m-1}(p-1)-p^{m-1}$&$p^{2m}(p^m-1)(p-1)$\\
             &\\[-6pt]
             $p^{2m-1}(p-1)$&$p(p^{2m}-1)$\\
             &\\[-6pt]
             $(p^{2m-1}+p^{m-1})(p-1)$&$p^{2m}(p^m-1)$\\
             &\\[-6pt]
             $p^{2m}$&$p-1$\\          
            \hline
       \end{tabular}
\end{table}

\begin{remark}
The ternary linear code defined by quadratic function $\Tr_{2m}(at^2+bt)+h$ over $\bF_{3^{2m}}$ in \cite{DC2} has parameters $[n,k,d]=[3^{2m},4m+1,2(3^{2m-1}-3^{m-1})]$. The ternary linear code(we used in this paper) defined by Hermitian function $\Tr_{2m}(at^{3^m+1}+bt)+h$ over $\bF_{3^{2m}}$ in $(\ref{Eqn_TC})$ has parameters $[n,k,d]=[3^{2m},3m+1,2\cdot 3^{2m-1}-3^{m-1}]$. Hence they are not equivalent to each other.
\end{remark}

To simplify notations, we write $f(t)$ as $(f(t))_{t\in \bF_{3^{2m}}}$ below.

\begin{thm}\label{thm_3.15}
The linear code $\cC_3(\bD_d(\cC(2m,3)))$ is spanned by 
\[\left\{\begin{matrix}\sum_{i=0}^{m-1}\Tr_{2m}(b_it^{(3^m+1)3^i+1})+\sum_{i=0}^{2m-1}\Tr_{2m}(b_i't^{3^i+1})\\+\sum_{i=0}^{m-1}\Tr_{m}(c_it^{(3^m+1)(3^i+1)})
+\Tr_{2m}(bt)+h:b,b_i,b_i'\in\bF_{3^{2m}},c_i\in\bF_{3^m},h\in\bF_3\end{matrix}\right\},\]
and it holds $2$-designs.
\end{thm}

\begin{thm}\label{thm_3.16}
The linear code $\cC_3(\bD_d(\cC(2m,3)))$ has length $n=p^{2m}$, dimension $k=\frac{9m^2+7m}{2}+1$ and the minimum distance lower bounded by $3^{2m-2}$.
\end{thm}

\begin{example}
The parameters of the linear code $\cC(2m,3)$ and $\cC_3(\bD_d(\cC(2m,3)))$ for $m=1,2$ are listed as follows:
\begin{displaymath}
\begin{matrix}
m&\cC(2m,3)&\cC_3(\bD_d(\cC(2m,3)))\\
1&[9,4,5]&[9,9,1]\\
2&[81,7,51]&[81,26,21].
\end{matrix}
\end{displaymath}

The linear code $\cC(4,3)$ has weight distribution\[1+1296z^{51}+240z^{54}+648z^{60}+2z^{81}.\]

The linear code $\cC_3(\bD_d(\cC(4,3)))$ has weight distribution
\begin{displaymath}
\begin{array}{llllll}
1&+648z^{21}&+240z^{27}&+38880z^{28}&+25920z^{29}\\
+104976z^{30}&+373248z^{31}&+678780z^{32}&+2491560z^{33}&+9305280z^{34}\\
+12791520z^{35}&+52067880z^{36}&+167585760z^{37}&+193771440z^{38}&+633582000z^{39}\\
+1789957440z^{40}&+1784204820z^{41}&+5114657520z^{42}&+12311494560z^{43}&+10655818920z^{44}\\
+26240268600z^{45}&+54869931360z^{46}&+40818498480z^{47}&+86821798860z^{48}&+155822087880z^{49}\\
+99765111888z^{50}&+181835828208z^{51}&+279785262240z^{52}&+153082363320z^{53}&+238171803600z^{54}\\
+311801503680z^{55}&+144740601000z^{56}&+190453223160z^{57}&+210148421760z^{58}&+81951931440z^{59}\\
+90132625584z^{60}&+82728913248z^{61}&+26672379840z^{62}&+24134094720z^{63}&+18117430380z^{64}\\
+4739847840z^{65}&+3450820320z^{66}&+2053913760z^{67}&+424174320z^{68}&+238097880z^{69}\\
+109483488z^{70}&+16715808z^{71}&+7076700z^{72}&+2442960z^{73}&+116640z^{74}\\
+58320z^{75}&+38880z^{77}&+6480z^{78}&+2106z^{80}&+2186z^{81}.
\end{array}
\end{displaymath}
\end{example}

\section{Proofs of the main results}\label{sec_4}
In this section, we prove Theorem \ref{thm_3.15} and Theorem \ref{thm_3.16}. We firstly state the generators of the linear code $\cC_3(\bD_d(\cC(2m,3)))$ for each integer $m\geq 2$, i.e. the rows of the incidence matrix  $M_{\bD_d}$ of $\bD_d(\cC(2m,3))$. Next, we simplify the form of these generators and give the proof of Theorem \ref{thm_3.15}. The results in Theorem $\ref{thm_3.15}$ imply that the linear code $\cC_3(\bD_d(\cC(2m,3)))$ is a subcode of the extended code of the $4$-th order generalized Reed-Muller code. It induces the lower bound of minimum weight of the code $\cC_3(\bD_d(\cC(2m,3)))$. Finally, we compute the dimension of the code $\cC_3(\bD_d(\cC(2m,3)))=\overline{\cC^{\perp}}^{\perp}$ by counting the number of elements in the defining set of $\cC$ and give the proof of Theorem \ref{thm_3.16}.

We can easily get the following lemma from the definition of $\cC_3(\bD_d(\cC(2m,3)))$ in Section \ref{sec_3}.

\begin{lemma}\label{lem_3.2}
The linear code $\cC_3(\bD_d(\cC(2m,3)))$ defined as above is generated by the vectors in the following set over $\bF_3$: \[\{(\Tr_{2m}(at^{3^m+1}+bt)+h)^2\ |\ a\in \bF_{3^m}^*,b\in \bF_{3^{2m}},h\in \bF_3\setminus\{\Tr_{2m}(b)\}\}.\]
\end{lemma}

Using Lemma \ref{lem_3.2}, for any $a\in \bF_{3^m}^*,b\in \bF_{3^{2m}},h\in \bF_3\setminus\{\Tr_{2m}(b)\}$, we have
\begin{equation}\label{Eqn_code1}
\begin{aligned}
(\Tr_{2m}(at^{3^m+1}+bt)+h)^2&=\Tr_{2m}(at^{3^m+1}+bt)^2+2h\Tr_{2m}(at^{3^{m}+1}+bt)+h^2\\
&=\Tr_{2m}(at^{3^m+1})^2+2\Tr_{2m}(at^{3^m+1})\Tr_{2m}(bt)+\Tr_{2m}(bt)^2\\
&+2h\Tr_{2m}(at^{3^m+1})+2h\Tr_{2m}(bt)+h^2\ \in \cC_3(\bD_d(\cC(2m,3))),
\end{aligned}
\end{equation}
\begin{equation}\label{Eqn_code2}
\begin{aligned}
(\Tr_{2m}(at^{3^m+1}-bt)+h)^2&=\Tr_{2m}(at^{3^m+1}-bt)^2+2h\Tr_{2m}(at^{3^{m}+1}-bt)+h^2\\
&=\Tr_{2m}(at^{3^m+1})^2-2\Tr_{2m}(at^{3^m+1})\Tr_{2m}(bt)+\Tr_{2m}(bt)^2\\
&+2h\Tr_{2m}(at^{3^m+1})-2h\Tr_{2m}(bt)+h^2\ \in \cC_3(\bD_d(\cC(2m,3))).
\end{aligned}
\end{equation}
Substracting $(\ref{Eqn_code1})$ from $(\ref{Eqn_code2})$,
\begin{equation}\label{Eqn_code3}
\Tr_{2m}(at^{3^m+1})\Tr_{2m}(bt)+h\Tr_{2m}(bt)\ \in \cC_3(\bD_d(\cC(2m,3))).
\end{equation}
Adding $(\ref{Eqn_code1})$ to $(\ref{Eqn_code2})$,
\begin{equation}\label{Eqn_code4}
2\Tr_{2m}(at^{3^m+1})^2+2\Tr_{2m}(bt)^2+h\Tr_{2m}(at^{3^m+1})+2h^2\ \in \cC_3(\bD_d(\cC(2m,3))).
\end{equation}

We now try to show that each addition item of $(\ref{Eqn_code1})$ is also in $\cC_3(\bD_d(\cC(2m,3)))$.

\begin{lemma}\label{lem_3.3}
For any $t\in\bF_{3^{2m}}$, the following equations hold
\begin{itemize}
	\item[(1)]$\sum_{a\in\bF_{3^m}^*}\Tr_{2m}(at^{3^{m}+1})^2=0$,
	\item[(2)]$\sum_{b\in\bF_{3^{2m}}^*}\Tr_{2m}(bt)^2=0$,
	\item[(3)]$\sum_{a\in\bF_{3^m}^*}\Tr_{2m}(at^{3^m+1})=0$.
\end{itemize}
\end{lemma}
\begin{proof}
It is obviously if $t=0$. For any $t\in\bF_{3^{2m}}$, the equation $(t^{3^m+1})^{3^m}=t^{3^m+1}$ implies that $t\in\bF_{3^{m}}$. Note that $\Tr_{2m}(a)=2\Tr_{m}(a)$ for any $a\in \bF_{3^m}$.

(1)\begin{align*}
\quad\sum_{a\in\bF_{3^m}^*}\Tr_{2m}(at^{3^m+1})^2&=\sum_{a\in\bF_{3^m}^*}\Tr_{m}(a)^2\\
    &=|\{a\in\bF_{3^m}^*|\Tr_{m}(a)\neq 0\}|(\textup{ mod }3)\\
    &= 3^m-|\{a\in\bF_{3^m}^*|\Tr_{m}(a)= 0\}|(\textup{ mod }3)\\
    &= 3^m- 3^{m-1}(\textup{ mod }3)\\
    &=0.
\end{align*}

(2)
\begin{align*}
\quad\sum_{b\in\bF_{3^{2m}}^*}\Tr_{2m}(bt)^2&=\sum_{b\in\bF_{3^{2m}}}\Tr_{2m}(bt)^2\\
    &=|\{b\in\bF_{3^{2m}}|\Tr_{2m}(b)\neq 0\}|(\textup{ mod }3)\\
    &=3^{2m}-3^{2m-1}(\textup{ mod }3)\\
    &=0.
\end{align*}

(3)
\begin{align*}
\quad\sum_{a\in\bF_{3^m}^*}\Tr_{2m}(at^{3^m+1})&=2\sum_{a\in\bF_{3^m}^*}\Tr_{m}(a)\\
&=2\Tr_{m}(\sum_{a\in\bF_{3^m}^*}a)\\
&=0.
\end{align*}\end{proof}
\begin{lemma}\label{lem_3.4}
The constant codeword $1\in \cC_3(\bD_d(\cC(2m,3)))$.
\end{lemma}
\begin{proof}
From Lemma \ref{lem_3.3} and (\ref{Eqn_code4}), for any $a\in\bF_{3^m}^*$, $b\in\bF_{3^{2m}}$, $h=\Tr_{2m}(b)+1$, we have
\begin{align*}
\sum_{b\in\bF_{3^{2m}}^*}&\sum_{a\in\bF_{3^m}^*}(2\Tr_{2m}(at^{3^m+1})^2+2\Tr_{2m}(bt)^2+h\Tr_{2m}(at^{3^m+1})+2h^2)\\
&=2\sum_{b\in\bF_{3^{2m}}^*}(\Tr_{2m}(b)+1)^2=\sum_{b\in\bF_{3^{2m}}^*}(\Tr_{2m}(b)+2)=1\in \cC_3(\bD_d(\cC(2m,3))).
\end{align*}
\end{proof}

\begin{lemma}\label{lem_3.5}
Let $a\in \bF_{3^m}^*$, $h\in \bF_3$ and\[H(a,h)=\{t\in \bF_{3^{2m}}|\Tr_{2m}(at^{3^m+1})+h=0\}.\] Then the set $\Delta(H(a,h))=\{t_1-t_2:t_1,t_2\in H(a,h)\}=\bF_{3^{2m}}$.
\end{lemma}
\begin{proof}
We need to show that for any $a\in \bF_{3^m}^*$ and $b\in \bF_{3^{2m}}$, the equation
\begin{equation}\label{Eqn_code55}
\begin{cases}
\textup{Tr}_{2m}(at^{3^m+1}+h)=0\\
\textup{Tr}_{2m}(a(t+b)^{3^m+1}+h)=0
\end{cases}
\end{equation}
has at least one solution $t\in \bF_{3^{2m}}$. Define $N(a,b,h)$ to be the number of solutions of $(\ref{Eqn_code55})$. Let $\chi,\chi'$ be canonical character of the additive group of $\mathbb{F}_{3^{2m}}$, $\mathbb{F}_{3}$, respectively. We denote $S(a,b,h)=\sum_{s_1,s_2\ne0}\chi’(s_1h+s_2h)\sum_{t\in\bF_{q}}\chi(as_1t^{3^m+1}+as_2(t+b)^{3^m+1})$ below.

If $h\ne 0$, then we have
\begin{align*}
3^2N(a,b,h)&=\sum_{t\in \bF_{3^{2m}}}\sum_{s_1,s_2\in\bF_3}\chi’\left\{s_1[\Tr_{2m}(at^{3^m+1})+h]+s_2[\Tr_{2m}(a(t+b)^{3^m+1})+h]\right\}\\
&=\sum_{s_1,s_2\in\bF_3}\chi’(s_1h+s_2h)\sum_{t\in\bF_{3^{2m}}}\chi(as_1t^{3^m+1}+as_2(t+b)^{3^m+1})\\
&=3^{2m}+2\sum_{s\ne0}\chi’(sh)\sum_{t\in\bF_{3^{2m}}}\chi(ast^{3^m+1})+S(a,b,h)\\
&=3^{2m}-2+2(3^m+1)\sum_{s\ne0}\chi’(sh)\sum_{t\in\bF_{3^m}^*}\chi(t)+S(a,b,h)\\
&=3^{2m}-2+2(3^m+1)+S(a,b,h).
\end{align*}
Note that $\chi$ acts nontrivially on $\bF_{3^m}$, otherwise, $\Tr_{2m}(t)=2\Tr_{m}(t)=0$ for any $t\in\bF_{3^m}$ which is a contradiction.  

If $h=0$, then we have 
\begin{align*}
3^2N(a,b,h)&=\sum_{t\in\bF_{q}}\sum_{s_1,s_2\in\bF_3}\chi(as_1t^{3^m+1}+as_2(t+b)^{3^m+1})\\
&=3^{2m}+2\sum_{s\ne0}\chi'(sh)\sum_{t\in\bF_q}\chi(ast^{3^m+1})+S(a,b,h)\\
&=3^{2m}-2+2(3^m+1)\sum_{s\ne0}\chi'(s)\sum_{t\in\bF_{3^m}^*}\chi(s)+S(a,b,h)\\
&=3^{2m}+2\cdot 3^m+S(a,b,h).
\end{align*}
From the Weil bound on exponential sums and $m\geq2$, we get the bound on $N(a,b,h)$.

In the case of $h\ne 0$, 
\[|3^2N(a,b,h)-3^{2m}+2-2(3^m+1)|\leq (3-1)^2\cdot 3^m\]
which induces that 
\begin{align*}
 N(a,b,h)\geq 2\cdot3^{m+1}+3^{m-2}\geq 1.
\end{align*}
In the case of $h=0$, 
\[|3^2N(a,b,h)-3^{2m}-2\cdot 3^{m}|\leq (3-1)^2\cdot3^m\]
which induces that
\begin{align*}
N(a,b,h)\geq 3^{2m-2}-3^{m}+3^{m-2}\geq 1.
\end{align*}
The result in this lemma then follows.
\end{proof}

\begin{lemma}\label{lem_3.6}
Let $b\in \bF_{3^{2m}}$. Then $\Tr_{2m}(bt)\in \cC_3(\bD_d(\cC(2m,3)))$ .
\end{lemma}
\begin{proof}
Let $a\in\bF_{3^m}^*$, $b\in\bF_{3^{2m}}$ and $h\in\bF_{3}\setminus\{\Tr_{2m}(b)\}$. From Lemma \ref{lem_3.5}, we have $\Delta(H(a,h))=\bF_{3^{2m}}$. Let $c_1\cdots,c_{2m}\in H(a,-h)$ be a basis of $\bF_{3^{2m}}$ over $\bF_3$. Submitting $c_1,\cdots,c_{2m}$ into $(\ref{Eqn_code3})$, 
\begin{align*}
(\Tr_{2m}(ac_i^{3^m+1})+h)\Tr_{2m}(bc_i)=2h\Tr_{2m}(bc_i)\in \cC_3(\bD_d(\cC(2m,3))).
\end{align*}
Since $h\in\bF_3\setminus\{\Tr_{2m}(b)\}$, we can choose $h\neq0$. Hence $\Tr_{2m}(bt)\in \cC_3(\bD_d(\cC(2m,3)))$. 
\end{proof}

\begin{lemma}\label{lem_3.7}
Let $b,b'\in\bF_{3^{2m}}$. Then $\Tr_{2m}(bt)\Tr_{2m}(b't)\in \cC_3(\bD_d(\cC(2m,3)))$.
\end{lemma}
\begin{proof}
Let $b_1=\frac{b+b'}{2}, b_2=\frac{b'-b}{2}\in\bF_{3^{2m}}$. Submitting $b_1,b_2$ into $(\ref{Eqn_code4})$,
\begin{equation}\label{Eqn_code5}
2\Tr_{2m}(at^{3^m+1})^2+2\Tr_{2m}(b_1t)^2+h_1\Tr_{2m}(at^{3^m+1})+2h_1^2\in \cC_3(\bD_d(\cC(2m,3))),
\end{equation}
\begin{equation}\label{Eqn_code6}
2\Tr_{2m}(at^{3^m+1})^2+2\Tr_{2m}(b_2t)^2+h_2\Tr_{2m}(at^{3^m+1})+2h_2^2\in \cC_3(\bD_d(\cC(2m,3))),
\end{equation}
where $h_1 \in \bF_3\setminus\{\Tr(b_1)\}$ and $h_2 \in \bF_3\setminus\{\Tr(b_2)\}$. We set $h_1=h_2=h$ for some $h\in\bF_3\setminus \{\Tr_{2m}(b_1t),\\ \Tr_{2m}(b_2t)\}$. Then subtracting $(\ref{Eqn_code6})$ from $(\ref{Eqn_code5})$, \[2\Tr_{2m}((b_1-b_2)t)\Tr_{2m}((b_1+b_2)t)\in \cC_3(\bD_d(\cC(2m,3))).\]
Hence $\Tr_{2m}(bt)\Tr_{2m}(b't)\in \cC_3(\bD_d(\cC(2m,3)))$.
\end{proof}

\begin{lemma}\label{lem_3.8}
Let $a,a'\in \bF_{3^m}$. Then $\{\Tr_{2m}(at^{3^m+1})$, $\Tr_{2m}(at^{3^m+1})\Tr_{2m}(a't^{3^m+1})\}\subseteq \cC_3(\bD_d(\cC(2m,3)))$.
\end{lemma}
\begin{proof}
By Lemma \ref{lem_3.4}, Lemma \ref{lem_3.7} and $(\ref{Eqn_code4})$,
\begin{equation}\label{Eqn_code7}
\Tr_{2m}(at^{3^m+1})^2+2h\Tr_{2m}(at^{3^m+1})\in \cC_3(\bD_d(\cC(2m,3))).
\end{equation} 
 We can choose $h_1\ne h_2\in\bF_3\setminus\{\Tr_{2m}(b)\}$ and submite $h_1,h_2$ into $(\ref{Eqn_code7})$,
 \begin{equation}\label{Eqn_code8}
\Tr_{2m}(at^{3^m+1})^2+2h_1\Tr_{2m}(at^{3^m+1})\in \cC_3(\bD_d(\cC(2m,3))),
\end{equation}
\begin{equation}\label{Eqn_code9}
\Tr_{2m}(at^{3^m+1})^2+2h_2\Tr_{2m}(at^{3^m+1})\in \cC_3(\bD_d(\cC(2m,3))).
 \end{equation}
 Subtracting $(\ref{Eqn_code9})$ from $(\ref{Eqn_code8})$, \[2(h_1-h_2)\Tr_{2m}(at^{3^m+1})\in \cC_3(\bD_d(\cC(2m,3))).\] We then get that $\Tr_{2m}(at^{3^m+1})\in \cC_3(\bD_d(\cC(2m,3)))$ and $\Tr_{2m}(at^{3^m+1})^2\in \cC_3(\bD_d(\cC(2m,3)))$. Let $a_1=\frac{a+a’}{2}$, $a_2=\frac{a-a’}{2}\in\bF_{3^m}$. Note that $\Tr_{2m}(a_1t^{3^m+1})^2, \Tr_{2m}(a_2t^{3^m+1})^2\in \cC_3(\bD_d(\cC(2m,3)))$. So we have  
\begin{align*}
\Tr_{2m}(a_1t^{3^m+1})^2-\Tr_{2m}(a_2t^{3^m+1})^2&=\Tr_{2m}((a_1+a_2)t^{3^m+1})\Tr_{2m}((a_1-a_2)t^{3^m+1})\\
&=\Tr_{2m}(at^{3^m+1})\Tr_{2m}(a't^{3^m+1})\in \cC_3(\bD_d(\cC(2m,3))).
\end{align*}
\end{proof}

\begin{lemma}\label{lem_3.9}
Let $a\in \bF_{3^m}^*$ and $b\in \bF_{3^{2m}}$. Then $\Tr_{2m}(at^{3^m+1})\Tr_{2m}(bt)\in \cC_3(\bD_d(\cC(2m,3)))$.
\end{lemma}
\begin{proof}
From Lemma \ref{lem_3.4}, Lemma \ref{lem_3.6}, Lemmma \ref{lem_3.7}, Lemma \ref{lem_3.8} and $(\ref{Eqn_code1})$, the result then follows.
\end{proof}

The following Lemma is easily obtained from the above results and $(\ref{Eqn_code1})$.
\begin{lemma}\label{3.10}
The linear code $\cC_3(\bD_d(\cC(2m,3)))$ over $\bF_3$ is generated by
\begin{equation}\label{Eqn_3.12}
\begin{Bmatrix}
\Tr_{2m}(bt)\Tr_{2m}(b't),\Tr_{2m}(bt),\Tr_{2m}(at^{3^m+1})\Tr_{2m}(bt),\\
\Tr_{2m}(at^{3^m+1})\Tr_{2m}(a't^{3^m+1}),\Tr_{2m}(at^{3^m+1}),\bm{1}\ |\ a,a'\in\bF_{3^m}^*,b,b'\in\bF_{3^{2m}}
\end{Bmatrix}.
\end{equation}
\end{lemma}

\begin{lemma}[See Cor.8.4 \cite{ZX}]\label{lem_3.11}
Let $p$ be a prime and $n$ be a positive integer. Let $t_1,\cdots,t_n\in \bF_{p^n}$. Then $\{t_1,\cdots,t_n\}$ is a basis of $\bF_{p^n}$ over $\bF_p$ if and only if 
\[
\begin{vmatrix}
t_1&t_2&\cdots&t_n\\
t_1^p&t_2^p&\cdots&t_n^p\\
\vdots&\vdots&\cdots&\vdots\\
t_1^{p^{n-1}}&t_2^{p^{n-1}}&\cdots&t_n^{p^{n-1}}
\end{vmatrix}\ne0.
\]
\end{lemma}

\begin{lemma}\label{lem_3.12}
The set \[\<\Tr_{2m}(bt)\Tr_{2m}(b't)\ |\ b,b'\in \bF_{3^{2m}}\>=\<\sum_{j=0}^{2m-1}\Tr_{2m}(c_jt^{p^j+1})\ |\ c_j\in\bF_{3^{2m}}\>.\]
\end{lemma}
\begin{proof}
For any $b,b'\in \bF_{3^{2m}}$, 
\begin{equation}
\begin{split}
\Tr_{2m}(bt)\Tr_{2m}(b't)&=\sum_{i=0}^{2m-1}\sum_{j=0}^{2m-1}b^{3^i}b'^{3^j}t^{3^i+3^j}\ \ \ =\sum_{i=0}^{2m-1}b^{3^i}(\sum_{j=0}^{2m-1}b'^{3^{j-i+2m}}t^{1+3^{j-i+2m}})^{3^i}\nonumber\\
&=\sum_{i=0}^{2m-1}b^{3^i}(\sum_{j=0}^{2m-1}b'^{3^{j}}t^{1+3^{j}})^{3^i}=\sum_{j=0}^{2m-1}\Tr_{2m}(bb'^{3^j}t^{1+3^{j}}).
\end{split}
\end{equation}
  Let $\xi$ be the primitive element of $\bF_{3^{2m}}$. Then $\{\xi,\xi^3,\cdots,\xi^{3^{2m-1}}\}$ forms a normal basis of $\bF_{3^{2m}}$ over $\bF_3$. By Lemma \ref{lem_3.11}, the elements in \[\{(b',b'^{3},\cdots,b'^{3^{2m-1}})\ |\ b'=\xi^{3^j},0\leq j\leq 2m-1\}\] are linear independently over $\bF_{3}$ which means that they form a basis of $\bF_{3^{2m}}^{2m}$ over $\bF_{3^{2m}}$. Hence $\<\sum_{j=0}^{2m-1}\Tr_{2m}(bb'^{3^j}t^{3^{j}+1})\ |\ b,b'\in \bF_{3^{2m}}\>=\<\sum_{j=0}^{2m-1}\Tr_{2m}(c_jt^{3^j+1})\ |\ c_j\in\bF_{3^{2m}}\>$. 
\end{proof}

\begin{lemma}\label{lem_3.13}
The set \[\<\Tr_{2m}(at^{3^m+1})\Tr_{2m}(bt)\ |\ a\in\bF_{3^m}^*, b\in \bF_{3^{2m}}\>=\<\sum_{i=0}^{m-1}\Tr_{2m}(c_it^{(3^m+1)3^i+1}\ |\ c_i\in\bF_{3^{2m}}\>.\]
\end{lemma}
\begin{proof}
For any $a\in\bF_{3^m}^*$ and $b\in \bF_{3^{2m}}$, 
\begin{align*}
\Tr_{2m}(at^{3^m+1})\Tr_{2m}(bt)&=2\sum_{i=0}^{m-1}(at^{3^m+1})^{3^i}\sum_{j=0}^{2m-1}(bt)^{3^j}=2\sum_{j=0}^{2m-1}(bt)^{3^{j}}\sum_{i=0}^{m-1}(at^{3^m+1})^{3^{i+j}}\\
&=2\sum_{j=0}^{2m-1}(bt\sum_{i=0}^{m-1}(at^{3^m+1})^{3^i})^{3^j}=2\sum_{i=0}^{m-1}\Tr_{2m}(ba^{3^i}t^{(3^m+1)3^i+1}).
\end{align*}
 If $a=0$ or $b=0$ then $\Tr_{2m}(at^{3^m+1})\Tr_{2m}(bt)=0$. We assume that $a\neq 0$ and $b\neq 0$. By Lemma \ref{lem_3.11}, the elements in \[\{(a,a^3,\cdots,a^{3^{m-1}})\ |\ a=\xi^{(3^m+1)3^j},0\leq j\leq m-1\}\]are linear independently over $\bF_3$ which means that they form a basis of $\bF_{3^{2m}}^m$ over $\bF_{3^{2m}}$. Hence $\<2\sum_{i=0}^{m-1}\Tr_{2m}(ba^{3^i}t^{(3^m+1)3^i+1})\ |\ a\in\bF_{3^m}, b\in \bF_{3^{2m}}\>=\<\sum_{i=0}^{m-1}\Tr_{2m}(c_it^{(3^m+1)3^i+1})\ |\ c_i\in\bF_{3^{2m}}\>$.
\end{proof}

\begin{lemma}\label{lem_3.14}
The set \[\<\Tr_{2m}(at^{3^m+1})\Tr_{2m}(a't^{3^m+1}):a,a'\in\bF_{3^m}\>=\<\sum_{i=0}^{m-1}\Tr_{m}(c_it^{(3^m+1)(3^j+1)}):c_i\in\bF_{3^m}\>.\]
\end{lemma}
\begin{proof}
For any $a,a'\in\bF_{3^m}$,
\begin{align*}
\Tr_{2m}(a't^{3^m+1})\Tr_{2m}(at^{3^m+1})&=\sum_{i=0}^{2m-1}(a't^{3^m+1})^{3^i}\sum_{j=0}^{2m-1}(at^{3^m+1})^{3^j}\ \ =\sum_{i=0}^{2m-1}((a't^{3^m+1})\sum_{j=0}^{2m-1}(at^{3^m+1})^{3^{2m+j-i}})^{3^i}\\
&=\sum_{i=0}^{2m-1}((a't^{3^m+1})\sum_{j=0}^{2m-1}(at^{3^m+1})^{3^{j}})^{3^i}=\sum_{i=0}^{2m-1}(\sum_{j=0}^{2m-1}a'a^{3^j}t^{(3^m+1)(3^j+1)})^{3^i}\\
&=\sum_{j=0}^{2m-1}\Tr_{2m}(a'a^{3^j}t^{(3^m+1)(3^j+1)})\ \ \ \ \ =\sum_{j=0}^{m-1}\Tr_{m}(a'a^{3^j}t^{(3^m+1)(3^j+1)}).\\
\end{align*}
Similar to the proof in Lemma \ref{lem_3.13}, we have $\<\sum_{j=0}^{m-1}\Tr_{m}(a'a^{3^j}t^{(3^m+1)(3^j+1)})\ |\ a,a'\in\bF_{3^m}\>=\<\sum_{j=0}^{m-1}\Tr_{m}(c_jt^{(3^m+1)(3^j+1)}):c_j\in\bF_{3^m}\>$.
\end{proof}

\proof[\textbf{Proof of Theorem \ref{thm_3.15}}]{
The first part of the theorem follows from Lemma \ref{3.10} and Lemma \ref{lem_3.12}$-$\ref{lem_3.14}.

Now we prove that the linear code $\cC_3(\bD_d(\cC(2m,3)))$ is affine invariant and it then holds $2$-designs by Theorem \ref{thm_2.10}. For any $\sigma_{s_1,s_2}\in \textup{GAut}(\cC)$ with $s_1\in\bF_q^*$ and $s_2\in\bF_q$, we only need to show that $\Tr_{2m}(b\sigma_{s_1,s_2}(t)+b’(\sigma_{s_1,s_2}(t))^{3^i+1}+b''(\sigma_{s_1,s_2}(t))^{(3^m+1)3^j+1}+c(\sigma_{s_1,s_2}(t))^{(3^m+1)(3^k+1)})+u\in \cC_3(\bD_d(\cC(m,3)))$ for all $0\leq i\leq 2m-1$, $0\leq j,k\leq m-1$, $b,b',b''\in \bF_q$, $c\in\bF_{3^m}$ and $u\in \bF_3$. It is easy to check that 
$\Tr_{2m}(b(s_1t+s_2)+b''(s_1t+s_2)^{3^j+1})+u\in \cC_3(\bD_d(\cC(m,3)))$. Then we get that
\begin{align*}
\Tr_{2m}((s_1t+s_2)^{(3^m+1)3^i+1})&=\Tr_{2m}((s_1^{3^{m+i}}t^{3^{m+i}}+s_2^{3^{m+i}})(s_1^{3^i}t^{3^i}+s_2^{3^i})(s_1t+s_2))\\
&=\Tr_{2m}((s_1t)^{3^{m+i}+3^i+1}+(s_1t)^{3^{m+i}+1}s_2^{3^i})\\
&+\Tr_{2m}((s_1t)^{3^i+1}s_2^{3^{m+i}}+s_1ts_2^{3^{m+i}+3^i}+(s_1t)^{3^{m}+1}s_2^{3^{2m-i}})\\
&+\Tr_{2m}(s_1ts_2^{3^m+3^{m-i}}+s_1ts_2^{3^{m}+3^{2m-i}}+s_2^{3^{m+i}+3^i+1})\in\cC_3(\bD_d(\cC(m,3))).
\end{align*}
Similarly, we have $\Tr_{2m}(c(s_1t+s_2)^{(3^m+1)(3^k+1)}))\in \cC_3(\bD_d(\cC(2m,3)))$. The result then follows.
$\hfill\square$}

Before continuing our calculations, we introduce some convenient terminology. For $0\leq k\leq 2$ and $0\leq j\leq 2m-1$, we denote the linear codes
\begin{align*}
\cC_{\gamma_{kj}}&=\{\sum_{i=0}^{n-1}\Tr_{2m}(a_i\gamma_{kj}^i)x^i| a_i\in\bF_q\},\\
\cC_{\gamma_{3j}}&=\{\sum_{i=0}^{n-1}\Tr_{2m}(a_3\gamma_{3j}^i)x^i| a_3\in\bF_{3^m}\},
\end{align*}
where $\gamma_{0j}=\xi$, $\gamma_{1j}=\xi^{(3^m+1)3^j+1}$, $\gamma_{2j}=\xi^{3^j+1}$,  and $\gamma_{3j}=\xi^{(3^j+1)(3^m+1)}$. We define linear code $\cC=\left \langle \cC_{\gamma_{kj}}: 0\leq k\leq 3,0\leq j\leq 2m-1 \right \rangle_{\bF_3}$.  For $0\leq k\leq 3$ and $0\leq j\leq 2m-1$, set
\begin{equation}\label{Eqn_Set}
\begin{split}
&S_{0j}=\{-3^i:0\leq i\leq 2m-1\},\\
&S_{1j}=\{-3^i(3^j(3^m+1)+1):0\leq i\leq 2m-1\},\\
&S_{2j}=\{-3^i(3^j+1):0\leq i\leq 2m-1\},\\
&S_{3j}=\{-3^i(3^j+1)(3^m+1):0\leq i\leq 2m-1\}.
\end{split}
\end{equation}

\begin{remark}
Note that $\cC_3(\bD_d(\cC(2m,3)))=\overline{\cC^{\perp}}^{\perp}$. By Theorem \ref{thm_2.5}, we know that each $\cC_{\gamma_{kj}}$, $0\leq k\leq 3,0\leq j\leq 2m-1$, has defining set $T_{kj}=\bF_q^*\setminus S_{jk}$.  Then $\cC$ has the defining set $T=\cap_{k=0}^3\cap_{j=0}^{n-1}T_{kj}$ by Proposition \ref{lem_2.2}. Let $S_k=\cup_{j=0}^{n-1} S_{kj}$ for $1\leq k\leq 3$. It is straightforward to verify that $S_{ij}=S_{ij'}$ or $S_{ij}\cap S_{ij'}=\varnothing$ and $S_i\cap S_{i'}=\varnothing$ for any $i,i',j,j'$ with $0\leq i\neq i'\leq 3$ and $0\leq j\neq j'\leq 2m-1$. So the defining set \[T=\cap_{k=0}^3\cap_{j=0}^{n-1}(\bF_q^*\setminus S_{jk})=\cap_{k=0}^3(\mathbb{F}_q^*\setminus(\cup_{j=0}^{n-1}S_{kj}))=\bF_q^*\setminus(\cup_{k=0}^3S_{k}).\]
\end{remark}

 Now we count the number of the elements in each $S_i$, $0\leq i\leq 3$. Rather than give a series of detailed proof, we outline them in the following lemmas.

\begin{lemma}\label{lem_3.17}
Let $S_{1j}$, $0\leq j\leq 2m-1$, be defined in $(\ref{Eqn_Set})$. Then we have 
\begin{itemize}
\item[(1)] For any $0\leq j\leq 2m-1$, $|S_{1j}|=2m$.
\item[(2)] For any $j_1,j_2$, $0\leq j_1\neq j_2\leq 2m-1$, $S_{1j_1}\cap S_{1j_2}=\varnothing$ or $S_{1j_1}=S_{1j_2}$ and the latter case holds if and only if $j_1\equiv j_2+m$ mod $2m$.
\item[(3)] $|S_1|=2m^2$.
\end{itemize}
\end{lemma}

\begin{lemma}\label{lem_3.18}
Let $S_{2j}$, $0\leq j\leq 2m-1$, be defined in $(\ref{Eqn_Set})$. Then we have
\begin{itemize}
\item[(1)] For any $0\leq j\leq 2m$, $|S_{2j}|=2m$ .
\item[(2)] For any $i_1,i_2,j_1,j_2$, $0\leq i_1,i_2, j_1\neq j_2\leq 2m-1$, $-3^{i_1}(3^{j_1}+1)=-3^{i_2}(3^{j_2}+1)$ if and only if $i_1\equiv i_2+j_2$ mod $2m$ and $j_1\equiv -j_2$ mod $2m$.
\item[(3)] $|S_2|=(2m+1)m$.
\end{itemize}
\end{lemma}

\begin{lemma}\label{lem_3.19}
Let $S_{3j}$, $0\leq j\leq 2m-1$, be defined in $(\ref{Eqn_Set})$. Then we have
\begin{itemize}
\item[(1)] For any $0\leq j\leq 2m-1$, 
\[|S_{3j}|=\begin{cases}\frac{m}{2}, \textup{ if }m \textup{ is even and }j=\frac{m}{2}\textup{ or }\frac{3m}{2}\\m,\textup{ otherwise}\end{cases}.\]
\item[(2)]  For any $i_1,i_2,j_1,j_2$, $0\leq i_1,i_2,j_1,j_2\leq 2m-1$ with $j_1\neq j_2$,  $-3^{i_1}(3^{j_1}+1)(3^m+1)=-3^{i_2}(3^{j_2}+1)(3^m+1)$ if and only if $i_1\equiv i_2$ mod $m$ and $j_1\equiv j_2$ mod $m$ or $i_1\not\equiv i_2$ mod $m$ and $j_1\equiv -j_2\equiv i_2$ mod $m$.
\item[(3)] $|S_3|=\frac{m(m+1)}{2}$.
\end{itemize}
\end{lemma}

\begin{remark}
Let $\cR_3(k,2m)$ be the extended code of generalized Reed-Muller code $\cR_3(k,2m)^*$. Note that the linear code $\cC_3(\bD_d(\cC(2m,3)))$ is a subcode of $\cR_3(4,2m)$ by Theorem \ref{thm_2.7}. 
\end{remark}

\proof[\textbf{Proof of Theorem \ref{thm_3.16}}]{
It is easy to check that $|S_0|=2m$. Then $|S|=\sum_{i=0}^3|S_i|=2m+2m^2+(2m+1)m+\frac{m(m+1)}{2}=\frac{9m^2+7m}{2}$ by Lemma \ref{lem_3.17}, Lemma \ref{lem_3.18} and Lemma \ref{lem_3.19}. Now we have dim$(\cC)=n-(n-|S|)=\frac{9m^2+7m}{2}$ by Theorem \ref{thm_2.3}. Then 
\begin{align*}
\textup{dim}(\cC_3(\bD_d(\cC(2m,3))))&=\textup{dim}(\overline{\cC^{\perp}}^{\perp})=\textup{dim}(\cC)+1=\frac{9m^2+7m}{2}+1.
\end{align*}
As stated before, the linear code $\cC_3(\bD_d(\cC(2m,3)))$ is a subcode of the code $\cR_3(4,2m)$ and $\cR_3(4,2m)$ has minimum distance $3^{2m-2}$ by Theorem \ref{thm_2.6}. It turns out that the minimum distance $d(\cC_3(\bD_d(\cC(2m,3))))$ is lower bounded by $3^{2m-2}$. $\hfill\square$}

\begin{remark}
The fourth-order generalized Reed-Muller code $\cR_3(4,2m)$ has dimension \[k=\binom{2m+3}{4}+\binom{2m+2}{3}-\frac{(2m-1)2m}{2}+1>\frac{9m^2+7m}{2}+1=\textup{dim}(\cC_3(\bD_d(\cC(2m,3)))).\]
\end{remark}

\section{Concluding Remark}\label{sec_5}
We computed the incidence matrices of $2$-designs that are supported by the minimum weight codewords of $\cC(2m,3)$. The linear code $\cC_3(\bD(\cC(2m,3)))$ with $m\geq 2$ generated by the rows of the incidence matrices has $\cC(2m,3)$ as its subcode and has many affine invariant subcodes. This means that the structure of these linear codes we obtained is richer than the previous one. We obtained that the linear code $\cC_3(\bD(\cC(2m,3)))$ is the subcode of the extended code of the $4$-th order generalized Reed-Muller code and gave the lower bound of the minimum weight of $\cC_3(\bD(\cC(2m,3)))$. We computed the dimension of the linear code $\cC_3(\bD(\cC(2m,3)))=\overline{\cC^{\perp}}^{\perp}$ by counting the number of elements in the defining set of $\cC$.

\section*{Acknowledgement}
The work of the first author was supported by National Natural Science Foundation of China under Grant No. 11771392.


\begin{thebibliography}{99}
\bibitem{As} Assmus Jr. E.F., Key J.D.: Polynomial codes and finite geometries. In: Pless V.S., Huffman W.C. (eds.)
The Handbook of Coding Theory, vol. II, pp. 1269–1343. Elsevier, Amsterdam (1998).\\
\bibitem{As0} Assmus Jr. E.F., Key J.D.: Designs and their codes. Cambridge University Press, Cambridge(1992).\\
\bibitem{As1} Assmus Jr.E.F., Mattson Jr.H.F.: New $5$-designs. J. Combinatorial Theory(A),6,122-151(1969).\\
\bibitem{DC1} Ding C.: Designs from linear codes. World Scientific, Singapore (2019).\\
\bibitem{DC2} Ding C., Tang C.,Tonchev D.: Linear codes of $2$-designs associated with subcodes of the ternary generalized Reed-Muller codes. Des.Codes Cryptogr. 88,625-641(2020).\\
\bibitem{DC3} Ding C., Tang C.: Infinite families of near MDS codes holding $t$-designs. IEEE Trans. Inf. Theory PP(99):1-1 (2020).\\
\bibitem{DC4} Ding C.: Infinite families of 3-designs from a type of five-weight code. Des. Codes Cryptogr., 86(3):703–719, 2018.\\
\bibitem{DC5} Ding C. and C. Li.: Infinite families of 2-designs and 3-designs from linear codes. Discrete Math., 340(10):2415–2431, 2017.\\
\bibitem{DX1} Du X., Wang R., Tang C., Wang Q.: Infinite families of $2$-designs from two classes of binary cyclic codes with three nonzeros. arXiv:1903.08153 [math.CO] (2019)\\
\bibitem{DX2} Du X.,Wang R.,Fan C.:Infinite families of $2$-designs from a class of cyclic codes with two non-zeros. arXiv:1904.04242 [math.CO] (2019).\\
\bibitem{HW} Huffman W.C., Pless V.: Fundamentals of error correcting codes. Cambridge University Press, Cambridg (2003).\\
\bibitem{TC} Tang C., Ding C.: An infinite family of linear codes supporting $4$-designs. arXiv:2001.00158 [math.CO] (2020).\\
\bibitem{WR} Wang R.,Du X.,Fan C.:Infinite families of $2$-designs from a class of non-binary Kasami cyclic codes. arXiv:1912.04745 [mathb.CO] (2019)\\
\bibitem{ZX} Zhe-Xian Wan: Finite fields and Galois rings. World Scientific, USA (2011).
\end{thebibliography}
\end{document}